\documentclass[11pt]{article}
\usepackage[margin=1in]{geometry}
\usepackage{verbatim,url,enumerate,color,paralist,graphicx}
\usepackage{amsmath,amsfonts,amssymb,amstext,amsthm,multicol,caption,subcaption}
\usepackage{fullpage}
\usepackage{hyperref}
\usepackage{todonotes}

\usepackage{algorithm,algpseudocode}

\newtheorem{theorem}{Theorem}[section]

\newtheorem{proposition}[theorem]{Proposition}

\newtheorem{lemma}[theorem]{Lemma}

\newtheorem{claim}{Claim}

\title{A 7/3-Approximation for Feedback Vertex Sets in Tournaments}

\author{
Matthias Mnich\thanks{Universit{\"a}t Bonn,  Bonn, Germany. Supported by ERC Starting Grant 306465 (BeyondWorstCase).}\\
  \texttt{mmnich@uni-bonn.de}
\and
Virginia Vassilevska Williams\thanks{Computer Science Department, Stanford University, USA. Supported by NSF Grants CCF-1417238, CCF-1528078 and CCF-1514339, and BSF Grant BSF:2012338.} \\
\texttt{virgi@cs.stanford.edu}
\and
L\'aszl\'o A. V\'egh\thanks{London School of Economics, UK. Supported by EPSRC First Grant EP/M02797X/1.} \\
\texttt{l.vegh@lse.ac.uk}
}

\date{}

\begin{document}

\maketitle
\thispagestyle{empty}

\begin{abstract}
  We consider the minimum-weight feedback vertex set problem in tournaments: given a tournament with non-negative vertex weights, remove a minimum-weight set of vertices that intersects all cycles.
  This problem is $\mathsf{NP}$-hard to solve exactly, and Unique Games-hard to approximate by a factor better than 2.
We present the first $7/3$ approximation algorithm for this problem, improving on the
previously best known ratio $5/2$ given by Cai et al. [FOCS 1998, SICOMP 2001].
\end{abstract}

\section{Introduction}
\label{sec:introduction}
Among the most basic concepts in graph theory is the notion of a \emph{feedback vertex set (FVS)} of a digraph: a subset of the vertices $S$ such that removing $S$ makes the digraph acyclic. The computational problem of finding a FVS of minimum size is known as the {\sc Feedback Vertex Set} problem. A fundamental problem with numerous applications (e.g. in deadlock recovery in operating systems), the {\sc Feedback Vertex Set} problem is among Karp's 21 original $\mathsf{NP}$-complete problems~\cite{Karp1972}. Karp's proof of NP-hardness also implies that the problem is APX-hard. Obtaining a constant factor polynomial-time approximation algorithm for the {\sc Feedback Vertex Set} problem seems elusive and is a major open problem. The best known approximation factor achievable in polynomial time is $O(\log n \log\log n)$~\cite{EvenEtAl1998,Seymour1995}.

The {\sc Feedback Vertex Set} problem is particularly interesting for the special case when the input graph is a {\em tournament}, i.e. an orientation of the complete graph. The problem restricted to tournaments has many interesting applications, most notably in social choice theory where it is essential to the definition of a certain type of election winners called the Banks set~\cite{banksset}.

The {\sc Feedback Vertex Set} problem remains $\mathsf{NP}$-complete and APX-hard in tournaments. Moreover, Speckenmeyer~\cite{Speckenmeyer1992} gave an approxi\-mation-ratio preserving polynomial time reduction from the {\sc Vertex Cover} problem in general undirected graphs to the {\sc Feedback Vertex Set} problem in tournaments. Consequently, the FVS problem in tournaments cannot be approximated in polynomial time within a factor better than $1.3606$ unless $\mathsf{P} = \mathsf{NP}$~\cite{DinurSafra2005}, and not within a factor better than 2 assuming the Unique Games Conjecture (UGC)~\cite{KhotRegev2008}.

On the upper bound side, the {\sc Feedback Vertex Set} problem in tournaments admits an easy $3$-approximation algorithm: while the tournament contains a directed triangle, place all the triangle vertices in the FVS and remove them from the tournament (see also Bar-Yehuda and Rawitz~\cite{BarYehudaRawitz2005} for another simple 3-approximation algorithm). 17 years ago,
Cai, Deng and Zang~\cite{CaiEtAl2001} improved the simple algorithm and gave a polynomial time algorithm with approximation guarantee $5/2$, even in the case when vertices have non-negative weights and one seeks a solution of approximate minimum weight.
%
In this paper we develop a better, $7/3$-approximation algorithm for the minimum weight {\sc Feedback Vertex Set} problem in tournaments, narrowing the gap to the UGC-based lower bound of~2 to $1/3$.

\begin{theorem}
\label{thm:main}
  There exists a polynomial time $7/3$-approximation algorithm for
  finding a minimum-weight feedback vertex set in a tournament.
\end{theorem}

In the process we uncover a structural theorem about tournament graphs that 
 has interesting connections to the tournament
colouring problem investigated by Berger et
al. \cite{BergerEtAl2013}. We explain these connections in Sect.~\ref{sec:hero}.

\paragraph{Overview.}

Let us first give an overview of Cai et al.'s result~\cite{CaiEtAl2001}.
Let ${\cal T}_5$ denote the set of tournaments on 5 vertices where the minimum FVS has size 2. 
Cai et al. showed that for any tournament free of subtournaments from ${\cal T}_5$, the minimum-weight FVS problem becomes
polynomial-time solvable.
They in fact show that the natural LP relaxation of the problem is integral in ${\cal T}_5$-free tournaments: the minimum
weight of a FVS equals the maximum value of a fractional directed triangle packing.

For the special case of unit weights only, their $5/2$-approximation algorithm starts by greedily choosing subtournaments
in ${\cal T}_5$, and including all 5 vertices in the FVS.
Once the remaining tournament admits no more subtournaments in ${\cal T}_5$, the optimal covering algorithm is used.
The algorithm returns an $5/2$-approximate optimal solution, since every step removing a subtournament decreases
the optimum value by at least 2, and includes 5  vertices in the FVS.
The algorithm extends to non-negative weights using the local ratio technique.

\medskip

We now give an overview of our approach. We define
the set ${\cal T}_7$ as the set of 7-vertex tournaments where the
minimum size of a FVS is 3. The algorithm comprises two stages. The
first stage uses the iterative rounding technique, and removes all
subtournaments in ${\cal T}_7$; the weight of the vertices included at this stage will be at
most $7/3$-times the decrease in the optimum weight. In the second
stage, we give a $7/3$-approximate combinatorial algorithm for the
remaining ${\cal T}_7$-free tournament.

The analogous first stage of Cai et al. obtains a worse factor
$5/2$.  In the second stage, their algorithm delivers an optimal
solution. In contrast, we only give an approximation algorithm in the
second stage, but that is
sufficient for the overall approximation guarantee. 

\medskip
We now provide some more detail of the two stages.
In the first stage we use the iterative rounding
technique.
We formulate the natural LP relaxation of the minimum-weight FVS problem in the given tournament $T$, including a
covering constraint for every directed triangle of $T$, and further we
include that every subtournament of $T$ belonging to~${\cal T}_7$ must be
covered by at least three vertices.
We consider an optimal solution of the LP relaxation. If there is a vertex of $T$ with fractional
value at least $3/7$, we include it in our FVS and remove it from~$T$. We then
resolve the LP  on the remaining tournament, and again include a
vertex with fractional value at least $3/7$, if there exists one. We iterate until there
are no more such vertices. At this point, the tournament will be ${\cal
T}_7$-free, and the fractional optimum value equals exactly one third
of the total weight of the vertices (see Lemma~\ref{lem:opt-frac}).


\medskip 

In the second stage, we develop a polynomial time combinatorial algorithm that delivers a
FVS of weight at most $7/9$ times the total weight of the vertices in
a ${\cal T}_7$-free tournament (Theorem~\ref{thm:T-7-free}).
Our algorithm implies our main theorem since an optimal  FVS in the remaining ${\cal T}_7$-free tournament is of size at least the optimum fractional value, which by the previous paragraph is exactly a third of the total weight of the nodes, which itself is at least $1/3\cdot 9/7 = 3/7$ of the size of the FVS returned.

To prove Theorem~\ref{thm:T-7-free}, we decompose the vertex set into ``layers''.
For the ${\cal T}_5$-free layers, we use Cai et al.'s algorithm as a subroutine to find an optimal solution to the FVS problem on such layers.
However, we cannot guarantee all layers to be $\mathcal T_5$-free, and thus include the ones that are not entirely in
the solution. The layering approach is inspired by Cai et al.'s
structural analysis of ${\cal T}_5$-free tournaments; nevertheless, we
use it quite differently.

\subsection{Related work}
Feedback vertex sets in tournaments are a well-studied subject.
Dom et al.~\cite{DomEtAl2010} showed how to decide existence of a feedback vertex set of size at most $k$ in time $2^k\cdot n^{O(1)}$, as well as a $O(k^2)$-sized kernel.
From an exact algorithms perspective, Gaspers and Mnich~\cite{GaspersMnich2013} showed how to compute a minimum FVS in time $O(1.674^n)$.
They further give the first polynomial-space algorithm to enumerate all minimal FVS of a given tournament with polynomial delay.
On the combinatorial side, they prove that any $n$-vertex tournament has at most $O(1.674^n)$ minimal FVS, thereby improving upon a long-standing bound by Moon~\cite{Moon1971} from 1971.

The related question of FVS
in {\em bipartite} tournaments has also been studied, i.e. orientations of the complete bipartite graph.
First, Cai, Deng and Zang~\cite{CaiEtAl2002} using a similar framework to their $5/2$-approximation algorithm \cite{CaiEtAl2001}, developed a $7/2$-approximation algorithm for FVS in bipartite tournaments. This was improved by Sasatte~\cite{Sasatte08} giving a $3$-approximation, and finally,
 by van Zuylen \cite{vanZuylen2011} who developed a polynomial time $2$-approximation algorithm. 

Iterative rounding is a standard and powerful method in approximation algorithms; we refer the reader to the book by Lau,
Ravi and Singh~\cite{LauEtAl2011}. The approach was made popular by
Jain's groundbreaking 2-approximation for survivable network design
\cite{Jain2001}, and the main application area is network design. However,
the same principle was already used earlier for various problems. In
particular, Krivelevich used implicitly iterative rounding for the undirected triangle cover
problem \cite{Krivelevich1995}; our application is similar to his argument.
Van Zuylen \cite{vanZuylen2011} used iterative rounding for FVS in bipartite tournaments.



\section{Description of the Algorithm}
\label{sec:notationandmainclaims}
Let $T=(V,A)$ be a tournament, equipped with a weight function $w:V\to \mathbb
R_{\geq 0}$. An arc between $u,v\in V$ will be denoted by $(u,v)\in A$ or $u\to v$.
The tournament~$T$ is {\em transitive} if it does not contain any
directed cycles, or equivalently, its vertices admit a topological order.
A vertex set $S\subseteq V$ is a {\em feedback vertex set} if
$T[V\setminus S]$ is transitive. For a vertex set $S\subseteq V$, let
$T-S$ denote the tournament resulting from the removal of the vertex set
$S$ from~$T$. If $S=\{v\}$ has a single element, we also use the
notation $T-v$.

The following straightforward characterization of FVS's in tournaments is well-known.
\begin{proposition}\label{prop:triangle}
  For any tournament $T$, a set $S$ is a feedback vertex set for $T$ if and only if $S$ intersects every directed triangle of $T$.
\end{proposition}

Let ${\cal T}_{5}$ denote the family of tournaments $T'$ on $5$ vertices that 
do not contain a transitive subtournament on $4$ vertices; equivalently,
every FVS of~$T'$ has size at least $2$. 
 The set ${\cal T}_5$ contains~3 tournaments,
the same ones used by Cai et al. Characterizations of many  related
classes of tournaments were given by Sanchez-Flores~\cite{SanchezFlores1998}.

Our main focus will be the set ${\cal T}_7$ defined as follows. Let ${\cal T}_7$ denote the family of tournaments on 7 vertices that do not contain a transitive subtournament on 5 vertices. This is equivalent to the property that  every FVS is of size at least 3.
We remark that $\mathcal T_7$ consists of 121 tournaments.
%
\begin{figure}[htb]
    \centering
    \begin{subfigure}{.5\textwidth}
      \centering
      \includegraphics[scale=1]{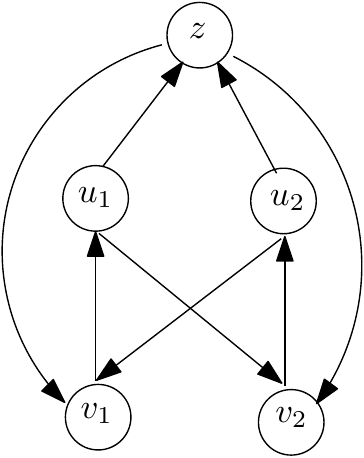}
      \caption{$S_5$}
    \end{subfigure}%
    \begin{subfigure}{.5\textwidth}
       \centering
       \includegraphics[scale=1]{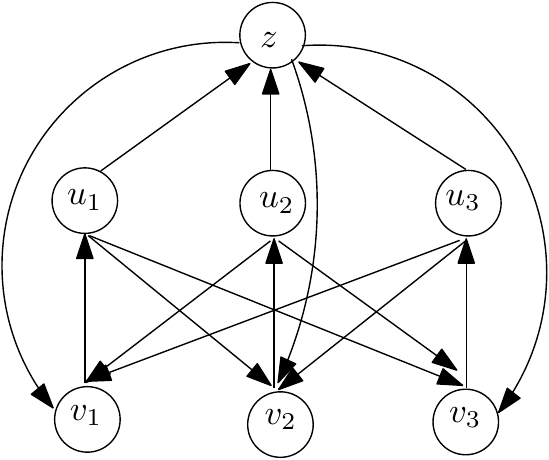}
       \caption{$S_7$}
    \end{subfigure}%
\caption{Examples of  ${\cal T}_5$ and  ${\cal T}_7$.}\label{fig:t5-t7}
  \end{figure}

Fig.~\ref{fig:t5-t7} gives important examples of tournaments  $S_5\in
{\cal T}_5$ and $S_7\in {\cal T}_7$. The arcs not included in the
figures can be oriented arbitrarily; hence both figures represent multiple
tournaments.
Tournament $S_5$ is identical to $F_1$ of Cai et al. \cite{CaiEtAl2001}. We leave the proof of the following simple claim to the reader.
\begin{proposition}
\label{cl:m3-plus-one}
For the tournaments in Fig.~\ref{fig:t5-t7}, $S_5\in {\cal T}_5$ and $S_7\in {\cal T}_7$.
\end{proposition}

For a tournament $T$, let $\Delta(T)$ denote the family of vertex sets
of directed triangles in~$T$. According to Propostion~\ref{prop:triangle}, $T$ is transitive if and
only if $\Delta(T)=\emptyset$. Similarly,  ${\cal T}_{5}(T)$ and   ${\cal T}_{7}(T)$ denote
the family of vertex sets of the subtournaments of $T$ isomorphic to a
tournament in ${\cal T}_{5}$ and~${\cal T}_{7}$, respectively.
We say that $T$ is \emph{$\mathcal T_5$-free} if
$\mathcal T_{5}(T)=\emptyset$ and \emph{$\mathcal T_7$-free} if $\mathcal T_{7}(T)=\emptyset$.

\medskip

We use iterative rounding for the following LP relaxation of the FVS problem in a
tournament $T=(V,A)$ with weight function $w:V\rightarrow \mathbb Q_{\geq 0}$:

\begin{equation}
  \begin{aligned}
    \min~& w^T x\\
    x(R)&\ge 1\quad \forall R\in \Delta(T)\\
    x(Q)&\ge 3\quad \forall Q\in {\cal T}_7(T) \\
    x&\ge 0
  \end{aligned}
\tag{LP}\label{eqn:lprelaxation}
\end{equation}

Notice that~\eqref{eqn:lprelaxation} does not impose any constraints
for subtournaments in $\mathcal T_5(T)$. This is an LP of polynomial
size. Let $OPT(T)$ denote the optimum value of (LP).

\renewcommand{\algorithmicrequire}{\textbf{Input:}}
\renewcommand{\algorithmicensure}{\textbf{Output:}}
\begin{algorithm}
  \begin{algorithmic}[1]
    \Require{A tournament $T = (V,A)$ with weight function $w:V\rightarrow\mathbb Q_{\geq 0}$.}
    \Ensure{A feedback vertex set of $T$ of weight at most $\frac 73 OPT(T)$.}
    \State Initialize $F = \emptyset$, $T'=T$.
      \State Find an optimal solution $x^*$ to \eqref{eqn:lprelaxation}.
    \While{$T'\neq\emptyset$ and there exists a vertex $v\in V(T')$ with $x^*_v \geq \frac 73$}
      \State Set $F:=F\cup\{v: x^*_v \geq \frac 73\}$ and $T':=T'\setminus \{v: x^*_v \geq \frac 73\}$.
      \State Remove every vertex from $T'$ not contained in any
      directed triangle; denote this resulting tournament also by $T'$.
      \State Solve \eqref{eqn:lprelaxation} for $T'$ to obtain an optimal
      solution $x^*$.
    \EndWhile
   \State If $T'\neq \emptyset$ then run Algorithm {\sc Layers}
   (Algorithm~\ref{alg:layers}) for $T'$, returning a FVS
   $F'$ of $T'$.\\
\Return $F\cup F'$.
 \end{algorithmic}
  \caption{{\sc Tournament FVS}}
\label{alg:wholealg}
\end{algorithm}

Our algorithm
(Algorithm~\ref{alg:wholealg}), iteratively builds a FVS $F$ of $T$, initialized empty. We
find an optimal solution $x^*$ to ~\eqref{eqn:lprelaxation}, and as long as there
exist vertices $v$ such that $x^*_v\ge \frac 37$, we include all of them in $F$ and remove them from $T$. We iterate this process, by resolving the
LP for the smaller tournament $T'$. 
By the first stage of the algorithm we mean the sequence of these
iterative rounding steps, which terminate once
 $T'$ becomes empty (in which case we are done), or every fractional
value $x^*_v$ satisfies $x^*_v < \frac 37$.

In this case, the current tournament $T'$ must be ${\cal
  T}_7$-free. Indeed, the constraint on the elements of ${\cal
  T}_7(T')$ guarantees that in every ${\cal T}_7$ subtournament
at least one element must have fractional value at least~$3/7$.
Note that this is true already after the very first iteration. 
The analogous task of removing all subtournaments from ${\cal T}_5(T')$ is done by Cai et al. \cite{CaiEtAl2001} using the local ratio technique.
As shown by Bar-Yehuda and Rawitz \cite{BarYehudaRawitz2005}, this could also be done via a primal-dual algorithm.
The local ratio and primal-dual techniques easily give a 3-approximation for the formulation with triangles only (given as \eqref{prog:primal} in the next Sect.). However, these do not seem to easily extend for our second goal with the iterative rounding, when we only have triangle constraints left, and we proceed as long as there is a vertex of fractional value at least $3/7$.

\medskip
In the second stage we apply Algorithm {\sc Layers}
(Algorithm~\ref{alg:layers}). That is the algorithm described in the
following theorem.

\begin{theorem}
\label{thm:T-7-free}
  There is an algorithm that, given any ${\cal T}_7$-free tournament $T'=(V',A')$ with weight function $w:
V'\rightarrow \mathbb Q_{\geq 0}$, in polynomial time finds a FVS $F'$ of $T'$ with weight at most $\frac 79w(V')$.
\end{theorem}

We defer the description of Algorithm {\sc Layers} as well as the proof of
Theorem~\ref{thm:T-7-free} to Sect.~\ref{sec:combinatorialproof}. 
We now prove the validity of Algorithm~\ref{alg:wholealg}, provided this
result.

\section{Proof of Theorem~\ref{thm:main}}
It is straightforward to see that the set $F\cup F'$ returned by the
algorithm is a FVS of $T$.
The next simple lemma shows that in every iterative rounding step, the weight of the elements
added to $F$ can be bounded by the decrease of $OPT(T)$.

\begin{lemma}\label{lem:it-round}
In every iteration during the first stage of the algorithm with
current tournament $T'$ and set $F$, we have
\begin{equation*}
w(F)\le \frac73(OPT(T)-OPT(T')) \enspace .
\end{equation*}
\end{lemma}
\begin{proof}
We prove the claim by induction.
 It is clearly true at the beginning when $T'=T$.
Whenever we remove a vertex not contained in any triangle, the left-hand side
remains unchanged and the right-hand side may only increase. It
is sufficient to prove that if $x^*$ is an optimal solution to 
\eqref{eqn:lprelaxation} for $T'$ and $S=\{v: x_v^*\ge
\frac37\}\neq\emptyset$,
then $OPT(T'\setminus S)+\frac 37w(S)\le OPT(T')$.

Note that $x^*$ restricted to $T'\setminus S$ is feasible to
\eqref{eqn:lprelaxation} for $T'\setminus S$, and thus
$OPT(T'\setminus S)\le OPT(T')-\sum_{v\in S}w(v)x^*_v\le OPT(T')-\frac37w(S)$, as
required. 
\end{proof}

As observed above, the tournament $T'$ at the end of the first stage
is ${\cal T}_7$-free. Theorem~\ref{thm:T-7-free} guarantees that the
FVS $F'$ of $T'$ returned by Algorithm {\sc Layers} has weight $w(F')\le \frac 79
w(T')$.
\begin{lemma}\label{lem:opt-frac}
If $\Delta(T')\neq\emptyset$ at the end of the first stage, then
$OPT(T')=\frac13 w(T')$. 
\end{lemma}
Before proving this lemma, let us see how it concludes the proof of Theorem~\ref{thm:main}.
According to Theorem~\ref{thm:T-7-free} and Lemma~\ref{lem:opt-frac}, $w(F') \leq \frac 79w(T') \leq \frac 73 OPT(T')$.
Using Lemma~\ref{lem:it-round}, we see that
the weight of the constructed FVS $F\cup F'$ is $w(F\cup F')\le \frac73 OPT(T)$.

The proof of Lemma~\ref{lem:opt-frac} analyzes the LP relaxation with triangle constraints only.
At the end of the first stage, $T'$ is $\mathcal{T}_7$-free. Hence the second set of constraints in
\eqref{eqn:lprelaxation} for $T'$ is empty. Let us omit these
constraints and write \eqref{eqn:lprelaxation}  together with its
dual:

\begin{multicols}{2}
  \noindent
  \begin{equation}\tag{P}\label{prog:primal}%
    \begin{aligned}%
    \min\ \ & w^T x\\
    x(R)&\ge 1\quad \forall R\in \Delta(T')\\
    x:~& V\to \mathbb R_+
    \end{aligned}
  \end{equation}
  \begin{equation}\tag{D}\label{prog:dual}%
    \begin{aligned}%
    \max \mathbf{1}^T y& \\
    \sum_{R: v\in R} y(R) &\le w_v \quad \forall v\in V'\\
    y:~& \Delta(T')\to \mathbb R_+
    \end{aligned}
  \end{equation}
\end{multicols}

\begin{proof}[Proof of Lemma~\ref{lem:opt-frac}]

We assumed that $\Delta(T')\neq \emptyset$, and that $x_v^*\le\frac 37$ for
every $v\in V'$, where $V'=V(T')$ is the vertex set of $T'$.
\begin{claim}
$x_v^*>0$ for every $v\in V'$.
\end{claim}
\begin{proof}
For a contradiction, assume $x_v^*=0$ for some $v\in V'$. Every vertex in
$T'$ is contained in a directed triangle; say $\{v,u,z\}\in
\Delta(T')$. The relaxation \eqref{eqn:lprelaxation} includes a constraint
$x^*_v+x^*_u+x^*_z\ge 1$, and therefore $x^*_u\ge \frac12$ or
$x^*_z\ge \frac 12$, a contradiction to $x^*_v \leq \frac{3}{7}$ for all $v\in V'$.
\end{proof}

  By primal-dual slackness, we must have $\sum_{u\in R} y(R)=w(u)$ for all $u\in V'$.
  Then
  \begin{equation*}
    w(V')=\sum_{u\in V'}\sum_{R:u\in R} y(R)=\sum_{R\in\Delta(T')}y(R)\sum_{u\in R}1=3 \sum_{R\in\Delta(T')}y(R)=3\cdot OPT(T'),
  \end{equation*}
completing the proof.
In the third equation, we used that every triangle contains exactly three vertices.
\end{proof}

\section{The Algorithm {\sc Layers}}
\label{sec:combinatorialproof}
In this section, we present Algorithm {\sc Layers} and prove Theorem~\ref{thm:T-7-free}.
First, we need the following result by Cai et al. \cite[Sect. 4]{CaiEtAl2001}.
\begin{theorem}[\cite{CaiEtAl2001}]\label{thm:chinese}
There exists an algorithm that, given any ${\cal T}_5$-free tournament $\hat{T}$ with non-negative vertex weights, finds in polynomial-time a minimum weight FVS in $\hat{T}$.
\end{theorem}
We shall refer to the algorithm as the \textsc{Cai-Deng-Zang} algorithm.
We also need a property of $\mathcal T_5$-free tournaments established by Cai et al.~\cite[Thm. 3.2]{CaiEtAl2001}.
\begin{proposition}[\cite{CaiEtAl2001}]
\label{thm:t5free-fractionalpackingcoveringequality}
  For any $\mathcal T_5$-free tournament $\hat{T}$ with non-negative vertex weights, the minimium weight of a FVS equals the maximum value of a fractional triangle packing.
\end{proposition}
The next simple lemma bounds the cost of the FVS found by the \textsc{Cai-Deng-Zang} algorithm  in terms of the total weight of the vertices $w(\hat V)$. 
\begin{lemma}\label{cl:chinese-bound}
Let $\hat T=(\hat V,\hat A)$ be a ${\cal T}_5$-free tournament with weight function $w:\hat{V}\rightarrow\mathbb Q_{\geq 0}$, and
let~$\hat F$ be an FVS of $\hat{T}$ returned by the \textsc{Cai-Deng-Zang} algorithm applied to $(\hat{T},w)$. Then
$w(\hat F)\le w(\hat V)/3$.
\end{lemma}
\begin{proof}
Consider \eqref{prog:primal} and \eqref{prog:dual} from the previous subsection.
  By Proposition~\ref{thm:t5free-fractionalpackingcoveringequality}, the polyhedron~\eqref{prog:primal} applied to $T' = \hat{T}$ and $w$ is integral. Consider an
optimal solution $y$ to \eqref{prog:dual}. Then
  \begin{equation*}
w(\hat F)=\mathbf{1}^T y=   \frac 13 \sum_{u\in \hat V}\sum_{R:u\in R}
y(R)\le \frac 13 \sum_{u\in \hat V}w(u)=w(\hat V)/3 \enspace . \qedhere
   \end{equation*}
\end{proof}

\subsection{Layers from a vertex}
Recall that Theorem~\ref{thm:T-7-free} takes as input a $\mathcal T_7$-free tournament $T' = (V',A')$ with weight function $w:V'\rightarrow\mathbb Q_{\geq 0}$.
For a set $S\subseteq V'$, let $N(S)=\{v\notin S~|~\exists u\in S, u\to
v\}$ denote the set of its in-neighbours; let $N(u):=N(\{u\})=\{v~|~
v\to u\}$.

For any vertex $z\in V'$ and $\ell\in\{1,\hdots,n\}$, let us define $V_\ell(z)$ as the set of vertices~$v$ such that the shortest directed path from $v$ to $z$ has length exactly $\ell-1$.
Equivalently, let $V_1(z)=\{z\}$, $V_2(z)=N(z)$, and for $\ell\geq 2$ let
\begin{equation*}
  V_{\ell+1}(z):=\{v\in V'\setminus (V_1(z)\cup\ldots\cup V_\ell(z))~|~ \exists u\in V_\ell(z), v\to u \} \enspace .
\end{equation*}

We will prove the following structural result.
For two disjoint sets $S,Z\subseteq V'$, let us say that $Z$ {\em in-dominates} $S$ if 
for every $s\in S$ there exists a $z\in Z$ with $s\to z$. We say that $Z$ \emph{$2$-in-dominates}~$S$ if
$Z$ has a subset $Z'\subseteq Z$ with $|Z'|\le 2$ such that $Z'$ in-dominates $S$. 
\begin{theorem}\label{thm:layers}
For every vertex $z$, the following hold:
\begin{enumerate}[(a)]
\item The set $V_3(z)$ is  ${\cal T}_5$-free, and is 2-in-dominated by $V_2(z)$.
\item The set $V_4(z)$ is  ${\cal T}_5$-free, and is 2-in-dominated by $V_3(z)$.
\item If $z$ is a minimum in-degree vertex in the tournament, then
  $V_2(z)$ is also   ${\cal T}_5$-free. 
\end{enumerate}
\end{theorem}
The proof of Theorem~\ref{thm:layers} is given in Sect.~\ref{sec:layer-proof}.
Let us now provide some context and motivation.
Cai et al. \cite{CaiEtAl2001} showed that for any ${\cal T}_5$-free tournament, if we select a minimum in-degree vertex $z$, then every layer $V_i(z)$ induces a transitive tournament and is 1-in-dominated by~$V_{i-1}(z)$. This is an important step in their algorithm for finding
the exact optimal solution in  ${\cal T}_5$-free tournaments.

Assume that the analogous property held for ${\cal T}_7$-free
tournaments $T'$: starting from a minimum in-degree vertex $z$, every
layer $V_{i-1}(z)$ is ${\cal T}_5$-free. Then one could get
a FVS of $T'$ with weight at most $\frac 23w(V')$ as follows. Compare the total
weight of the even and odd layers, and include in the FVS whichever of
the two is smaller. Let us assume the total weight of the odd
layers is smaller; the argument is same for the other case. For every remaining even layer $V_{i}(z)$, run  the
\textsc{Cai-Deng-Zang} algorithm to obtain a FVS $F_{i}$ of $V_i(z)$. Form
the final FVS $F'$ of $T'$ as the union of all odd layers and the union of the $F_i$'s for the
even layers. Using
Proposition~\ref{cl:chinese-bound},  it is easy to verify $w(F')\le \frac 23 w(V')$. 
Further, $F'$ will be a FVS of $T'$, since by the construction of the $V_i(z)$'s,
every triangle must fall on consecutive layers.

However, Theorem~\ref{thm:layers} only claims ${\cal T}_5$-freeness
of layers $V_i(z)$ for $i\le 4$. 
This property might not hold for higher values of $i$.
To overcome this difficulty, we modify the layering procedure. While the layers
are constructed, we already include certain vertices in the final
FVS. This is to make sure that for every layer $U_i$, it holds that
$U_i=V_j(z')$ in some subtournament of $T$, for a certain vertex~$z'$ in a previous layer and $j=3$
or $j=4$. Hence Theorem~\ref{thm:layers} guarantees that all the
constructed layers are ${\cal T}_5$-free. The construction of the
final FVS
will be a modification of the simple argument above.

\subsection{Description of the layering algorithm}

\renewcommand{\algorithmicrequire}{\textbf{Input:}}
\renewcommand{\algorithmicensure}{\textbf{Output:}}
\begin{algorithm}[htb]
  \begin{algorithmic}[1]
    \Require{A ${\cal T}_7$-free tournament $T' = (V',A')$ with weight function $w:V'\rightarrow\mathbb Q_{\geq 0}$.}
    \Ensure{A feedback vertex set $F'$ of $T'$ of weight at most $\frac 79 w(V')$.}
    \State Choose $z_1$ as a vertex of minimum in-degree.
    \State Set $U_1:=\{z_1\}$, 
    \State Set $U_2:=N(z_1)$,  $W:=V'\setminus (U_1\cup U_2)$, $k:=1$.
    \While{$W\neq \emptyset$}
    \State Set $U_{2k+1}:=N(U_{2k})\cap W$, $W:=W\setminus U_{2k}$.
    \State Set $U':=N(U_{2k+1})\cap W$, $W:=W\setminus U'$.
    \State Choose $z_{2k+1}\in U_{2k+1}$ such that $w(U'\cap N(z_{2k+1}))\ge w(U')/2$.
    \State Set $U_{2k+2}:=U'\cap N(z_{2k+1})$; $S_{2k+2}:=U'\setminus N(z_{2k+1})$.
    \State Set $k:=k+1$.
    \EndWhile
    \State Set $L_0:=\cup_{j=1}^{k} U_{2j}$,
    $L_1:=\cup_{j=0}^{k-1} U_{2j+1}$, and $S:=\cup_{j=1}^{k} S_{2j}$.
    \If{$w(L_0)\ge w(L_1)$}
    \State Run the  \textsc{Cai-Deng-Zang} algorithm for every
    $U_{2j}$ to obtain a FVS $F_{2j}$ of $U_{2j}$.
    \State Set $F':=(\cup_{j=1}^{k} F_{2j})\cup S\cup L_1$.
     \Else
  \State Run the  \textsc{Cai-Deng-Zang} algorithm for every
    $U_{2j+1}$ to obtain a FVS $F_{2j+1}$ of $U_{2j+1}$.
     \State Set $F':=(\cup_{j=0}^{k-1} F_{2j+1})\cup S\cup L_0$.
\EndIf\\
\Return $F'$.
 \end{algorithmic}
  \caption{{\sc Layers}}
\label{alg:layers}
\end{algorithm}
The algorithm (Algorithm~\ref{alg:layers}) first
partitions the vertex set $V'$ into $S\cup \bigcup_{j=1}^{2k} U_j$ for
some $2k\le n$.
We now describe how the layers are constructed in Steps 1-11.
We start by setting
$U_1=\{z_1\}$ for a vertex $z_1$ of minimum in-degree. We let $U_2=N(z_1)$
be the set of in-neighbours of~$z_1$. The set $W$ will denote the set of
vertices not yet included in some $U_k$ or in $S$; at this point,
$W=V'\setminus (U_1\cup U_2)$.

While $W$ is not empty, we construct an odd
layer $U_{2k+1}$, an even layer $U_{2k}$, and $S_{2k+1}$ as follows. Set $U_{2k+1}$ will be simply the set of in-neighbours of
$U_{2k}$ inside $W$; we remove~$U_{2k+1}$ from~$W$. 
In the remaining set $W:= W\setminus U_{2k+1}$, let $U'$ be the set of in-neighbours of
$U_{2k+1}$. We partition~$U'$ into $U_{2k+2}$ and $S_{2k+2}$, and
remove $U'$ from $W$. To obtain this partitioning, we
 pick a vertex $z_{2k+1}\in U_{2k+1}$ such that $w(N(z_{2k+1})\cap U')\ge
w(U')/2$. The existence of such a vertex~$z_{2k+1}$ is non-trivial, and will be
proved in Lemma~\ref{lem:Ui}(c). We set $U_{2k+2}=N(z_{2k+1})\cap U'$, and
$S_{2k+2}=U'\setminus U_{2k+2}$; the set $S_{2k+2}$ will be part of $S$.

The layering procedure finishes once $W=\emptyset$. At this point, we denote by 
$L_0=\bigcup_{j=1}^k U_{2j}$ the set of all even and by
$L_1=\bigcup_{j=0}^{k-1} U_{2j+1}$ the set of all odd layers, and by
$S=\bigcup_{j=1}^k S_{2j}$ the set of vertices removed during the
procedure. Thus,
$V'=S\cup L_0\cup L_1$.
Given the layering, the algorithm constructs a FVS
in Steps 12-18 as follows. If $w(L_0)\ge
w(L_1)$, then we use the \textsc{Cai-Deng-Zang} algorithm  to find an optimal
FVS $F_{2j}$ in all even layers $U_{2j}$. We set the entire FVS as
 $F':=(\cup_{j=1}^k F_{2j})\cup S\cup L_1$.
Otherwise, we use the  \textsc{Cai-Deng-Zang} algorithm in all odd
layers to find optimal FVS's $F_{2j+1}$, and set $F':=(\cup_{j=0}^{k-1} F_{2j+1})\cup S\cup L_0$.

\subsection{Proof of correctness}
The following lemma summarizes the essential properties of the
layering obtained.
\begin{lemma}\label{lem:Ui}
The sets $S$ and $U_i$ returned by Algorithm~\textsc{Layers} satisfy the following properties.
\begin{enumerate}[(a)]
 \item If $i>j+1$, then $u\to v$ for every $u\in U_j$ and $v\in U_i$.
\item Every subtournament $T'[U_i]$ is ${\cal T}_5$-free.
\item  There always exists a vertex $z_{2i+1}\in U_{2i+1}$ as required in line 7 of the algorithm.
\item $w(S)\le w(L_0)$.
\end{enumerate}
\end{lemma}
\begin{proof}
Part (a) is immediate, since if $v\in U_j$, then $N(v)\subseteq
\cup_{\ell=0}^{j+1} (U_\ell\cup S_\ell)$ (let us use the convention
$S_{\ell}=\emptyset$ for all odd values of $\ell$). 

We prove parts (b) and (c) simultaneously. Part (b) is a direct consequence of Theorem~\ref{thm:layers} for
layers $1\le i\le 4$, as $z_1$ was chosen as a minimum in-degree
vertex. The existence of vertex $z_3\in U_3$ follows by
Theorem~\ref{thm:layers}(c): at this point, $U_3=V_3(z_1)$,
$U'=V_4(z_1)$, and thus~$U'$ is 2-in-dominated by $U_3$. This means that
there exist $z,z'\in U_3$ such that $N(z)\cup N(z')\supseteq U'$.
Without loss of generality, we may assume $w(U'\cap N(z))\ge w(U'\cap N(z'))$. Then
$z_3=z$ gives an appropriate choice.

Let us apply  Theorem~\ref{thm:layers} in the tournament $T''$ that is
the restriction of $T'$ to the ground set $\{z_3\}\cup U_4\cup U_5\cup
U_6$. In $T''$ we have $V_3(z_3)=U_5$ and $V_4(z_3)=U_6$, and therefore~$U_5$ and~$U_6$ are both ${\cal T}_5$-free. Further, $U_6$ is 2-in-dominated
by $U_5$ and therefore we can choose an appropriate $z_5\in U_5$ as
above. The same argument works for all values of $i\ge 3$: consider
the restriction of $T'$ to $\{z_{2i-1}\}\cup U_{2i}\cup U_{2i+1}\cup
U_{2i+2}$, and apply  Theorem~\ref{thm:layers}. We obtain that~$U_{2i+1}$ and $U_{2i+2}$ are ${\cal T}_5$-free as well as the choice
of $z_{2i+1}\in U_{2i+1}$.

Finally, part (d) is straightforward, since $w(U_{2i+2})\ge w(S_{2i+2})$
by the choice of $z_{2i+2}$.
\end{proof}

We are ready to prove the correctness and approximation ratio of the algorithm.
\begin{proof}[Proof of Theorem~\ref{thm:T-7-free}]
By Lemma~\ref{lem:Ui}(b), the  \textsc{Cai-Deng-Zang} algorithm can be
applied in all layers $U_i$ and finds an optimal FVS $F_i$ in
polynomial time.

First, let us show that the set $F'$ returned by Algorithm~\textsc{Layers}
is indeed a FVS of $T'$. For a contradiction, assume $V'\setminus F'$ contains a
directed triangle $uvs$. 

Let us assume $w(L_0)\ge w(L_1)$; the other case follows similarly.
In this case, $V'\setminus F'\subseteq L_0$.
The three vertices $u,v$ and $s$ cannot fall into the same layer
$U_{2i}$, as in every such layer we removed a FVS $F_{2i}$. Hence 
they must fall into at least two different $U_{2i}$'s. By
Lemma~\ref{lem:Ui}(b), if vertices fall into different even layers,
then all arcs from the lower layers point towards the higher layers,
excluding the possibility of such a triangle.

The proof is complete by showing $w(F')\le \frac 79w(V')$, or equivalently, 
$w(V'\setminus F')\ge \frac 29w(V')$.

\smallskip

\noindent{\bf Case I:  $w(L_0)\ge w(L_1)$}. In this case,
$w(V'\setminus F')=\cup_{j=1}^k (U_{2j}\setminus F_{2j})$. By
Proposition~\ref{cl:chinese-bound}, $w(F_{2j})\le w(U_{2j})/3$ for all
layers, and thus $w(V'\setminus F')\ge \frac 23 w(L_0)$.
Using Lemma~\ref{lem:Ui}(c), $w(L_0)\ge \max\{ w(L_1),w(S)\}$, and therefore
$w(L_0)\ge w(V')/3$. Thus $w(V'\setminus F')\ge \frac 29 w(V')$ follows.

\smallskip

\noindent{\bf Case II:  $w(L_0)< w(L_1)$}. Using the same argument as
in the previous case, we obtain $w(V'\setminus F')\ge \frac 23 w(L_1)$.
Again using Lemma~\ref{lem:Ui}(c), $w(L_1)>w(L_0)\ge w(S)$, and therefore
$w(L_1)\ge w(V')/3$, implying $w(V'\setminus F')\ge \frac 29 w(V')$.
\end{proof}

\subsection{Proof of Theorem~\ref{thm:layers}}\label{sec:layer-proof}
Let us first verify part (c): 
\begin{lemma}\label{lem:V_2}
Let $z$ be a minimum in-degree
vertex in a ${\cal T}_7$-free tournament. Then $V_2(z)$ is ${\cal T}_5$-free.
\end{lemma}
\begin{proof}
 We first claim that for every $u\in V_2(z)$ there must exist a $v\in V_3(z)$ with $v\to u$.
 Indeed, assume that for some $u$ there exists no such $v$. Then
 $N(u)\subsetneq V_2(z)=N(z)$ must hold. This is a contradiction to
 the choice of $z$ with $|N(z)|$ minimum.

  Consider a subset $H\subseteq V_3(z)$ containing at least one vertex $v$ with $v\to u$ for every $u\in V_2$; choose $H$ minimal for containment. 
  If $|H|\ge 3$, then there must be three vertices  $v_1,v_2,v_3\in
  H$, and three vertices $u_1,u_2,u_3\in V_2(z)$ such that $v_i\to
  u_i$ for $i=1,2,3$, while $u_i\to v_j$ if $i\neq j$. Then~$z$ and
  these vertices together form an $S_7\in {\cal T}_7$ subtournament as in
  Fig.~\ref{fig:t5-t7}(b), a contradiction.

  Hence $|H|\le 2$.
  For a contradiction, assume $X\subseteq V_2(z)$ forms a ${\cal T}_5$-graph ($|X|=5$).
  There exists a $v\in H$ with $|\{s\in X: v\to s\}|\ge 3$.
  We claim that $X\cup\{v,z\}\in {\cal T}_7$.
  Indeed, assume it contains a transitive tournament $Y$ on 5 vertices.
  Since $X\in {\cal T}_5$, $|X\cap Y|\le 3$; hence $v,z\in Y$ and  $|X\cap Y|=3$.
  There must be a vertex $t\in X\cap Y$ with $v\to t$, and thus~$vtz$ is a directed triangle, a contradiction.
\end{proof}

For parts (a) and (b) of Theorem~\ref{thm:layers}, we show that the 2-in-domination
claim implies ${\cal T}_7$-freeness:

 \begin{lemma}
\label{cl:2-parents-suffice}
Let $z$ be an arbitrary vertex in a ${\cal T}_7$-free tournament.
  For $i\ge 3$, if $V_i(z)$ is 2-in-dominated by $V_{i-1}(z)$, then
  $V_i(z)$ is ${\cal T}_5$-free.
\end{lemma}
\begin{proof}
  Consider a ${\cal T}_5$-subtournament $X$ in $V_i(z)$. By 2-in-domination, there must be
  a $v\in V_{i-1}(z)$ such that $|N(v)\cap X|\ge 3$.
 Let $s\in V_{i-2}(z)$ be such that $v\to s$.
  We obtain a contradiction as in the previous proof showing that $X\cup \{v,s\}\in {\cal T}_7$.
\end{proof}

The proof of Theorem~\ref{thm:layers} is complete by the following two lemmata, that show that both~$V_3(z)$ and $V_4(z)$ are 2-in-dominated by the previous layer.
\begin{lemma}
\label{cl:2-parents}
For an arbitrary vertex $z$ in a ${\cal T}_7$-free tournament $T'$, the set $V_3(z)$ is 2-in-dominated by $V_2(z)$.
\end{lemma}
\begin{proof}
  Let $H\subseteq V_2(z)$ be a minimal set for containment that
  in-dominates $V_3(z)$.   We show that $|H|\le 2$.
  Indeed, if $|H|\ge 3$, then again there must be a tournament $S_7$ as in Fig.~\ref{fig:t5-t7}(b), formed by~$z$, three vertices in $V_2(z)$ and three in $V_3(z)$.
\end{proof}

In the sequel, let  $\{a,b\}\subseteq V_2(z)$ be a set that
2-in-dominates $V_3(z)$.

\begin{lemma}
\label{lem:v4}
 For an arbitrary vertex $z$ in a ${\cal T}_7$-free tournament $T'$, the set $V_4(z)$ is 2-in-dominated by $V_3(z)$.
\end{lemma}
\begin{proof}
  For the sake of contradiction, assume that the
  minimal set in $V_3(z)$ 2-in-dominating $V_4(z)$ has size at least 3. Then there must exists vertices $u_1,u_2,u_3\in V_3(z)$ and
  $v_1,v_2,v_3\in V_4(z)$ such that $v_i\to
  u_i$ for $i=1,2,3$, while $u_i\to v_j$ if $i\neq j$. 

  If all $u_1,u_2,u_3\in N(a)$, then
  $\{a,u_1,u_2,u_3,v_1,v_2,v_3\}$ forms an $S_5$ tournament, a
  contradiction.  A similar argument applies for $b$.
  We may therefore assume (by possibly renaming the indices) that
  $u_1\to a, u_2\to a, u_3\to b, a\to u_3$, and $b\to u_2$.
  See Figure~\ref{fig:new}.

 \begin{figure}[htb]
  \centering
     \includegraphics[width=0.25\textwidth]{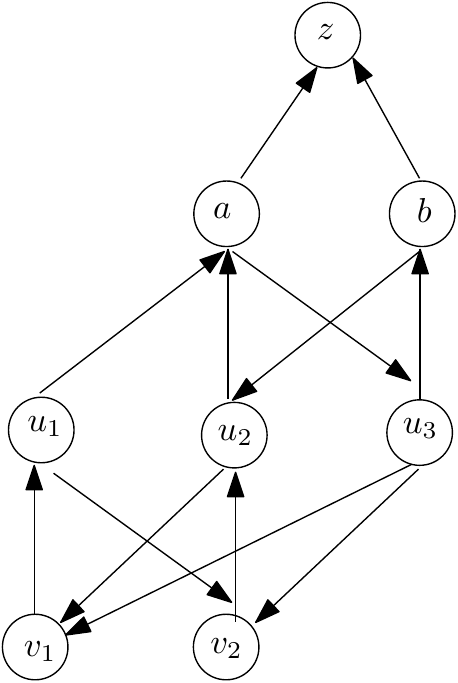}
\caption{Illustration of the proof of Lemma~\ref{lem:v4}. A few directed edges that are not portrayed are: from $z$ to each one of $\{u_1,u_2,u_3,v_1,v_2\}$ and from each of $\{a,b\}$ to each of $\{v_1,v_2\}$.}\label{fig:new}
  \end{figure}

 Since $T'$ is ${\cal T}_7$-free, then every 7 vertex subgraph of
 $\{z,a,b,u_1,u_2,u_3,v_1,v_2,v_3\}$ must contain a transitive
 tournament on $5$ vertices.
 Let $Q=\{z,a,b,u_2,u_3\}$, and  for $i=1,2$, let $Q_i$ denote $Q\cup \{u_1,v_i\}$. Let $T_i$ be a transitive tournament on $5$ nodes in $Q_i$. 
 
 Notice that that $Q$ forms a ${\cal T}_5$. Because of this, for both $i=1$ and $i=2$, $\{u_1,v_i\}\subseteq T_i$.
 Furthermore, $b,u_3$ and $z$ cannot all be in $T_i$ since they form a directed triangle; so 
 $\{a,u_2\}\cap T_i\neq \emptyset$.  A symmetric argument shows that  $\{b,u_3\}\cap T_i\neq\emptyset$ as well.

Now, since either $u_2\to u_1$ or $u_1\to u_2$, either $u_2u_1v_2$ or $u_1u_2v_1$ forms a directed triangle. Thus, $u_2\notin T_i$ for either $i=1$ or $i=2$. 
For the same $i$, $a\in T_i$ because of  $\{a,u_2\}\cap T_i\neq \emptyset$.
 Then $z$ cannot be in $T_i$ because $u_1az$ forms a directed triangle. Hence $T_i=\{a,b,u_1,u_3,v_i\}$, and this implies that
\begin{enumerate}[(i)]
\item $a\to b$ since $a\rightarrow u_3\rightarrow b$, 
\item $u_1\to u_3$ since $u_1\rightarrow a\rightarrow u_3$, and
\item $u_1\to b$  since $u_1\rightarrow a\rightarrow b$, using (i).
\end{enumerate}

As noted above, $\{u_1,v_1\}\subseteq T_1$. By (ii), $v_1u_1u_3$ forms a directed triangle, and by (iii), $v_1u_1b$ forms a triangle. Hence, neither $u_3$ nor $b$ can be contained in $T_1$, contradicting that  $\{b,u_3\}\cap T_1\neq\emptyset$. This completes the proof of Lemma~\ref{lem:v4}.
\end{proof}

\section{Connections to Tournament Colouring}\label{sec:hero}
We explore a connection to the notion of heroes and celebrities in
tournaments studied by  Berger, Choromanski,
Chudnovsky, Fox, Loebl, Scott, Seymour and
Thomass{\'e}~\cite{BergerEtAl2013}. Colouring a tournament means
partitioning its vertex set into transitive subtournaments; the
chromatic number of a tournament is the minimum number of colours needed.
A tournament $H$ is called a {\em hero}, if there exists a constant $c_H$
such that every $H$-free tournament has chromatic number at most $c_H$. Further, $H$ is
called a \emph{celebrity}, if for some constant $c'_H$, every $H$-free
tournament $T$ has a transitive subtournament of size at least $c'_H|V(T)|$.
Clearly, every hero is a celebrity; Berger et al. show that the
converse also holds: every celebrity is a hero.
Their work gives a characterization of all tournaments that are heros
(or equivalently, celebrities).

In this context, our Theorem~\ref{thm:T-7-free} shows that ${\cal
  T}_7$ collectively form a celebrity set. Further, our constant
$c'=2/9$ seems much better than the constants that could be derived
using the techniques of Berger et al.~\cite{BergerEtAl2013}. The set ${\cal T}_7$
includes some heros as well as some non-hero tournaments.
In contrast, the set ${\cal T}_5$ is precisely the set of heros on 5
vertices.

Berger et al.'s~\cite{BergerEtAl2013} characterization rules out the
following possible modification of our algorithm to obtain a 2-approximation for the {\sc Feedback Vertex Set} in tournaments problem. Instead of ${\cal T}_7$, one could use the
single tournament $ST_6$, the unique 6-vertex tournament not
containing a transitive subtournament of order 4~\cite{SanchezFlores1998}.
All copies $ST_6$ can be removed from the input tournament by losing a factor 2 in the approximation ratio only (instead of losing $7/3$ by removing copies of subtournaments from $\mathcal T_7$).
However, according to Berger et al.~\cite[Thm. 1.2]{BergerEtAl2013}, $ST_6$ is
\emph{not} a hero, and hence  there is no hope to prove a version of 
Theorem~\ref{thm:T-7-free} for this setting.

{
\bibliographystyle{abbrv}
\bibliography{fvstbib}
}

\end{document}